\documentclass[draftcls,11pt]{IEEEtran}
\onecolumn 
\usepackage[cmex10]{amsmath}
\usepackage{xcolor}
\usepackage[multiple]{footmisc}

\usepackage{mdwtab}
\usepackage{eqparbox}
\usepackage{tabularx}
\usepackage{graphicx,amssymb,rangecite,upgreek,graphicx,dsfont,mathrsfs}

\usepackage{algorithm}
\usepackage{algorithmic} 
\usepackage[usenames,dvipsnames]{pstricks}
 \usepackage{epsfig}

 \usepackage{pst-grad} 
 \usepackage{pst-plot} 

\newcommand {\bxt} {\mbox{\footnotesize\boldmath $x$}}

\newcommand {\balphat} {\mbox{\footnotesize\boldmath $\alpha$}}
\newcommand {\byt} {\mbox{\footnotesize\boldmath $y$}}

\newcommand {\bu} {\mbox{\boldmath $u$}}
\newcommand {\bv} {\mbox{\boldmath $v$}}
\newcommand {\bw} {\mbox{\boldmath $w$}}
\newcommand {\bx} {\mbox{\boldmath $x$}}
\newcommand {\by} {\mbox{\boldmath $y$}}
\newcommand {\bz} {\mbox{\boldmath $z$}}

\newcommand {\bA} {\mbox{\boldmath $A$}}
\newcommand {\bB} {\mbox{\boldmath $B$}}

\newcommand {\bE} {\mathbb{E}}
\newcommand {\pr} {\mathbb{P}}

\newcommand {\bF} {\mbox{\boldmath $F$}}

\newcommand {\bH} {\mbox{\boldmath $H$}}

\newcommand {\bP} {\mbox{\boldmath $P$}}

\newcommand {\bR} {\mbox{\boldmath $R$}}

\newcommand {\bZ} {\mbox{\boldmath $Z$}}

\newcommand {\bXt} {\mbox{\boldmath \footnotesize $X$}}
\newcommand {\bYt} {\mbox{\boldmath \footnotesize $Y$}}

\newcommand{\calA}{{\cal A}}

\newcommand{\calC}{{\cal C}}

\newcommand{\calF}{{\cal F}}
\newcommand{\calG}{{\cal G}}

\newcommand{\calI}{{\cal I}}

\newcommand{\calS}{{\cal S}}
\newcommand{\calT}{{\cal T}}

\newcommand{\define}{\stackrel{\triangle}{=}}

\def\btheta{{\mbox{\boldmath $\theta$}}}
\newcommand {\bthet} {\mbox{\footnotesize\boldmath $\btheta$}}

\def\btheta{{\mbox{\boldmath $\theta$}}}

\def\balpha{{\mbox{\boldmath $\alpha$}}}

\def\thetavecsc{{\mbox{\boldmath \tiny $\theta$}}}

\newcommand{\be}{\begin{equation}}
\newcommand{\ee}{\end{equation}}
\newcommand{\beqna}{\begin{eqnarray}}
\newcommand{\eeqna}{\end{eqnarray}}

\usepackage{theorem}
\DeclareFontFamily{U}{mathx}{\hyphenchar\font45}
\DeclareFontShape{U}{mathx}{m}{n}{
      <5> <6> <7> <8> <9> <10>
      <10.95> <12> <14.4> <17.28> <20.74> <24.88>
      mathx10
      }{}
\DeclareSymbolFont{mathx}{U}{mathx}{m}{n}
\DeclareMathSymbol{\bigtimes}{1}{mathx}{"91}
\theorembodyfont{\rmfamily}

\newcommand{\abs}[1]{\left|#1\right|}

\DeclareMathOperator{\re}{Re}
\DeclareMathOperator{\Img}{Im}
\DeclareMathOperator{\tr}{tr}
\theoremheaderfont{\itshape}

\newtheorem{theorem}{Theorem}
\newtheorem{proof}{Proof}

\newtheorem{lemma}{Lemma}

\newtheorem{remark}{Remark}
\newcommand{\p}[1]{\left(#1\right)}
\newcommand{\pp}[1]{\left[#1\right]}
\newcommand{\ppp}[1]{\left\{#1\right\}}
\newcommand{\norm}[1]{\left\|#1\right\|}

\begin{document}

\title{Universal Decoding for Gaussian Intersymbol Interference Channels}
\author{Wasim~Huleihel
        and~Neri~Merhav
				\\
        Department of Electrical Engineering \\
Technion - Israel Institute of Technology \\
Haifa 32000, ISRAEL\\
E-mail: \{wh@tx, merhav@ee\}.technion.ac.il
\thanks{$^\ast$This research was partially supported by The Israeli Science Foundation (ISF), grant no. 412/12.}
}
\date{}
\maketitle
\thispagestyle{empty}

\IEEEpeerreviewmaketitle

\begin{abstract}
\boldmath A universal decoding procedure is proposed for the intersymbol interference (ISI) Gaussian channels. The universality of the proposed decoder is in the sense of being independent of the various channel parameters, and at the same time, attaining the same random coding error exponent as the optimal maximum-likelihood (ML) decoder, which utilizes full knowledge of these unknown parameters. The proposed decoding rule can be regarded as a frequency domain version of the universal maximum mutual information (MMI) decoder. Contrary to previously suggested universal decoders for ISI channels, our proposed decoding metric can easily be evaluated.
\end{abstract}


\section{Introduction}

\IEEEPARstart{I}{n} many practical situations encountered in coded communication systems, the specific channel over which transmission is to be carried out is unknown to the receiver. The receiver only knows that the channel belongs to a given family of channels. In such a case, the implementation of the optimum maximum likelihood (ML) decoder is precluded, and thus, universal decoders, independent of the unknown channel, are sought. In designing such a decoder, there are two desirable properties that should be taken into account: The first is that the universal decoder performs asymptotically as well as the ML decoder had the channel law been known, and secondly, that the constructed decoding metric will be reasonably easy to calculate. This paper addresses the problem of universal decoding for intersymbol interference (ISI) Gaussian channels. 

The topic of universal coding and decoding under channel uncertainty has received very much attention in the last four decades, see, for example, \cite{NeriUni,Goppa,CsisKro,Csis2,ZivUni,LapZiv,FerderLapidoth,merFeder,Lomnitz,Lomnitz2,Misra,Shayevitz,Shayevitz2,Shayevitz3,UniNeri2}. In the realm of memoryless channels, Goppa \cite{Goppa} explored the maximum mutual information (MMI) decoder, which chooses the codeword having the maximum empirical mutual information (MMI) with the channel output sequence. It was shown that this decoder achieves the capacity in the case of discrete memoryless channels (DMC). In \cite{CsisKro}, the problem of universal decoding for DMC's with finite input and output alphabets was studied. It was shown that the MMI decoder universally achieves the optimal random coding error exponent under the uniform random coding distribution over a certain type class. In \cite{NeriUni}, an analogous result was derived for a certain parametric class of memoryless Gaussian channels with an unknown deterministic interference signal. In the same paper, a conjecture was proposed concerning a universal decoder for ISI channels. 

For channels with memory, there are several quite general results, each proposing a different universal decoder. In \cite{ZivUni}, the case of unknown finite-state channels with finite input and output alphabets for which the next channel state is a deterministic unknown function of the channel current state and current inputs and outputs, was considered. For uniform random codes over a given set, a universal decoder (that achieves the optimal random coding error exponent) which is based on the Lempel-Ziv algorithm was proposed. Later, in \cite{LapZiv}, it was shown that this decoder continues to be universally asymptotically optimum also for the class of finite-state channels with stochastic, rather than deterministic, next-state functions. In \cite{FerderLapidoth}, sufficient conditions and a universal decoder (called the merging decoder) were proposed, for families of channels with memory. The idea was to employ many decoding lists in parallel, each one corresponding to one point in a dense grid (whose size grows with the input block length) in the index set. Accordingly, with regard to our work, it was shown that the proposed decoder universally achieves the optimal error exponent under the ISI channel. Unfortunately, as was mentioned before, this deocder is very hard to implement in practice due to its implicit structure and the fact that it requires to form a dense grid in the parameter space. In \cite{merFeder}, a competitive minimax criterion was proposed. According to this approach, an optimum decoder is sought in the quest for minimizing (over all decision rules) the maximum (over all channels in the family) ratio between the error probability associated with a given channel and a given decision rule, and the error probability of the ML decoder for that channel, possibly raised some power less than unity. This decoder is, again, very hard to implement for the ISI channel due its complicated decoding metric. 

In this paper, we propose a universal decoder that asymptotically achieves the optimal error exponent, and contrary to previous proposed decoders, our proposed decoding metric can easily be calculated. The technique used in this paper is in line with the techniques which were established in \cite{NeriUni,Wasim}. Specifically, similarly to \cite{NeriUni}, the main idea is to define an auxiliary ``backward channel", which is a mathematical tool for assessing log-volumes of conditional typical sets of sequences with continuous-valued components. These log-volume terms play a pivotal role in the universal decoding metric. The backward channel is defined in a way that guarantees two properties: first, a measure concentration property, that is, assignment of high probability to a given conditional type by an appropriate choice of certain parameters, and secondly, the conditional density of the input given the output, associated with this backward channel should depend on the input and the output only via the sufficient statistics that define the conditional type class. Contrary to the problem considered in \cite{NeriUni}, the difficulty, in the ISI channel, stems from the fact that the choice of the backward channel is a non-trivial issue. It turns out that in this case, the passage to the frequency domain resolves this difficulty. The proposed decoding rule can be regarded as a frequency domain version of the universal maximum mutual information (MMI) decoder. 

The remaining part of this paper is organized as follows. In Section \ref{sec:model}, we first present the model and formulate the problem. Then, the main results are provided and discussed. In Section \ref{sec:proofout}, we provide a proof outline where we discuss the techniques and methodologies that are utilized in order to prove the main result. Finally, in Section \ref{sec:proof}, the main results are proved.

\allowdisplaybreaks

\section{Model Formulation and Main Result}\label{sec:model}
Consider a discrete time, Gaussian channel characterized by
\begin{align}
y_t = \sum_{i=0}^kh_ix_{t-i}+w_t,\ \ t=0,1,2,\ldots,n
\label{GenModel}
\end{align}
where $\ppp{x_t}$ are the channel inputs, $\ppp{h_i}_{i=0}^k$ is the unknown channel impulse response, $\ppp{w_t}$ is zero-mean Gaussian white noise with an unknown variance $\sigma^2>0$, and $\ppp{y_t}$ are the channel outputs. It will be assumed that the noise $\ppp{w_t}$ is statistically independent of the input $\ppp{x_t}$. We allow $k$ to grow with $n$ in the order of $k=o\p{n^{1/2}}$. In such a case, we further assume that the impulse response sequence $\ppp{h_i}_{i=0}^\infty$ is absolutely summable\footnote{This assumption can be relaxed to square summability of $\ppp{h_i}$.}. 

The input is a codeword that is randomly and uniformly drawn over a codebook $\calC = \ppp{\bx^1,\ldots,\bx^M}$ of $M=2^{nR}$ messages $\bx^i = \p{x_1^i,\ldots,x_n^i}\in\mathbb{R}^n$, $i=1,2,\ldots,M$, where $R$ is the coding rate in bits per channel use. In the following, the probability of error associated with the ML decoder, that knows the unknown parameters $\p{\sigma^2,h_0,\ldots,h_k}$, will be denoted by $P_{e,o}\p{\calC,R,n}$. We shall adopt the random coding approach, where each codeword is randomly chosen with respect to a probability measure denoted by $\mu\p{\bx}$. For a given power constraint, a reasonable choice of $\mu\p{\cdot}$ is the truncated Gaussian density restricted to the shell of an $n$-dimensional hypersphere whose radius is about $n\sigma_x^2$. To wit,
\begin{align}
\mu\p{\bx} = \nu^{-1}\psi_{\Delta}\p{\bx}\prod_{t=0}^{n-1}\exp\ppp{-\frac{x_t^2}{2\sigma_x^2}}
\label{inputMeasure}
\end{align}
where $\psi_{\Delta}\p{\bx}$ is the indicator function of the set 
\begin{align}
D_\Delta\define\ppp{\bx:\;\abs{\frac{1}{n}\sum_{t=0}^{n-1}x_t^2-\sigma^2_x}\leq\Delta\sigma_x^2}
\end{align}
where $\Delta\ll1$, and $\nu$ normalizes the above measure such that it would integrate to unity. Note that $\mu\p{\bx}$ is invariant to unitary transformations of $\bx$. It is well-known \cite[Chap. 7]{Gallager} that $\mu\p{\cdot}$ attains a higher error exponent than that of the respective Gaussian density with the same variance, at least for small rates, where the non-typical events (or, the large deviations events) are the dominant\footnote{Intuitively speaking, this is true because of the fact that it does not allow low energy codewords}. The analysis in this paper can also be carried for the case where the codewords are drawn independently and uniformly over a set $\calI_n\subseteq\mathbb{R}^n$ that is endowed with a $\sigma$-algebra (e.g., an $n$-dimensional hypercube), and satisfy an average power constraint, as was considered in \cite[Theorem 4]{FerderLapidoth}. Let $\bar{P}_{e,o}\p{R,n}\define\bE\ppp{P_{e,0}\p{\calC,R,n}}$, where the expectation is taken over the ensemble of randomly selected codebooks under $\mu\p{\cdot}$. Finally, we define the random coding error exponent as $E\p{R}\define-\limsup_{n\to\infty}n^{-1}\log\bar{P}_{e,o}\p{R,n}$. 

As was mentioned previously, we wish to find a decoding procedure which is universal in the sense of being independent of the unknown parameters, and at the same time attaining $E\p{R}$. Specifically, let $P_{e,u}\p{\calC,R,n}$ designate the error probability associated with the universal rule for a given codebook $\calC$, and let $\bar{P}_{e,u}\p{R,n}\define\bE\ppp{P_{e,u}\p{\calC,R,n}}$. Then, we would like $\bar{P}_{e,u}\p{R,n}$ to decay exponentially with rate $E\p{R}$. 

We now turn to present the proposed decoding rule. To this end, let $\tilde{\bx}$ and $\tilde{\by}$ denote the discrete Fourier transforms (DFT) of the sequences $\ppp{x_t}$ and $\ppp{y_t}$, respectively, i.e., the $m$-th component of $\tilde{\bx}$ is given by
\begin{align}
\tilde{x}_m = \frac{1}{\sqrt{n}}\sum_{t=0}^{n-1}x_te^{-j2\pi mt/n}
\end{align}
where $j=\sqrt{-1}$ and similarly for $\tilde{\by}$. Then, define an auxiliary ``backward channel" by the conditional measure
\begin{align}
V\p{\tilde{\bx}\vert\tilde{\by},\btheta,k} = \prod_{m=0}^{n-1}\p{2\pi\sigma_0^2}^{-1/2}\exp\ppp{-\frac{1}{2\sigma_0^2}\abs{\tilde{x}_m-\tilde{y}_m\sum_{l=0}^k\alpha_le^{\frac{2\pi jlm}{n}}}^2}
\label{BackwardChannel}
\end{align}
where $\btheta\define\p{\sigma_0^2,\alpha_0,\ldots,\alpha_k}$ is the parameters vector of the backward channel, in which $\ppp{\alpha_l}_{l=0}^k$ are complex-valued. It should be emphasized that the above definition of the auxiliary backward channel is completely unrelated to the underlying probabilistic model. In particular, it is not argued that $V\p{\tilde{\bx}\vert\tilde{\by},\btheta,k}$ is obtained from $\mu\p{\bx}$ and the forward channel \eqref{GenModel} by the Bayes rule, or any other relationship. For example, our backward channel allows vectors $\bx$ that are outside the region $D_\Delta$. Our decoding rule will select a message $\tilde{\bx}^i$ that maximizes the metric
\begin{align}
u\p{\tilde{\bx}^i,\tilde{\by}} = \frac{\max_{\bthet}V\p{\tilde{\bx}^i\vert\tilde{\by},\btheta,k}}{\mu\p{\tilde{\bx}^i}}
\label{risk}
\end{align}
among all $M$ codewords. The backward channel is a mathematical tool for assessing log-volumes of typical sets \cite{NeriUni,Wasim,Wasim2}, and it should be defined in a way that guarantees two general properties: first, a measure concentration property, that is, assignment of high probability to a given conditional type by an appropriate choice of the parameters of this backward channel, and secondly, the conditional density of $\tilde{\bx}$ given $\tilde{\by}$, associated with the backward channel should depend on $\tilde{\bx}$ and $\tilde{\by}$ only via the sufficient statistics that define the conditional type class. Contrary to the problem considered in \cite{NeriUni}, the difficulty in the ISI channel stems from the fact that the choice of the backward channel is a non-trivial issue. Specifically, as will be seen in the sequel, an ``appropriate" candidate backward channel must depend on a sufficient statistics vector (associated with $\bx$) with dimension that equals to the number of degrees of freedom, which in turn adjust their conditional expectations. It turns out that in this case, the passage to the frequency domain is more ``natural" and mathematically convenient due to the well-known asymptotic spectral properties of Toeplitz matrices (see, for example, \cite{Gray}). To wit, it can be seen that the model in \eqref{GenModel} can be written in a vector form $\by = \bA\bx+\bw$ where $\bA = \ppp{a_{i,j}} = \ppp{h_{i-j}}$ is a Toeplitz matrix. Now, by the spectral decomposition theorem \cite{Widom}, we know that there exists an orthonormal basis that diagonalizes the matrix $\bA$. Projecting the observations onto this basis will simply decompose the original channel into a set of independent channels, which are simpler to analyze. While this is true for any matrix $\bA$, for Toeplitz matrices we can asymptotically characterize their eigenvalues and eigenvectors in terms of the generating sequence $\ppp{h_i}_{i=0}^k$, which is a fundamental part in our analysis. We next give the main result of this paper. 

\begin{theorem}\label{th:1}
Let the codewords of $\calC$ be chosen randomly and independently with respect to the density $\mu\p{\cdot}$ given in \eqref{inputMeasure}. Assume that the channel impulse response coefficients are absolutely summable $\ppp{h_l}_{l=1}^\infty\in\ell_1$, and that $k=o\p{n^{1/2}}$. Then,
\begin{align}
\limsup_{n\to\infty}\frac{1}{n}\pp{\log \bar{P}_{e,u}\p{R,n} - \bar{P}_{e,0}\p{R,n}}\leq\xi\p{\Delta}
\end{align}
where $\xi\p{\Delta}\to0$ as $\Delta\to0$, and $\bar{P}_{e,u}\p{R,n}$ is the average probability of error associated with the universal decoder given in \eqref{risk}.
\end{theorem}

The intuitive interpretation of \eqref{risk} is that $n^{-1}\log u\p{\tilde{\bx},\tilde{\by}} = n^{-1}\log\max_{\bthet}V\p{\tilde{\bx}\vert\tilde{\by},\btheta,k}/\mu\p{\tilde{\bx}}$ is an empirical version of the per-letter mutual information between $\bx$ and $\by$ in the frequency domain. Thus, we select the input $\tilde{\bx}$ that seems empirically ``most dependent" upon the given output vector $\tilde{\by}$ in the frequency domain, which corresponds to the MMI principle. The passage to the frequency domain asymptotically eliminates the strong interactions between the various components of the input vector, and transforms the original model into a set of $n$ separable channels which are controlled by $\p{k+2}$ degrees of freedom. Note that on the support of $\mu\p{\cdot}$, the term $n^{-1}\log \mu\p{\tilde{\bx}^i}$ is nearly a constant independent of $i$. Thus, the proposed decoding rule is essentially equivalent to one that maximizes $\max_{\bthet}V\p{\tilde{\bx}\vert\tilde{\by},\btheta,k}$, namely, maximum a posteriori (MAP) decoding.

\begin{remark}
In \cite{NeriUni}, a universal decoding procedure for memoryless Gaussian channels with a deterministic interference was proposed. Accordingly, we remark that Theorem \ref{th:1} can be fairly easily extended to the channel model
\begin{align}
y_t = \sum_{i=0}^kh_ix_{t-i}+z_t+w_t
\end{align}
where $\ppp{z_t}$ is an unknown deterministic interference that can be decomposed as a series expansion of orthonormal bounded functions with an absolutely summable coefficient sequence, namely,
\begin{align}
z_t = \sum_{i=1}^\infty b_i\phi_{i,t}, \ \ t=1,2,\ldots
\end{align}
where $\ppp{b_i}\in\ell_1$ and $\abs{\phi_{i,t}}\leq L<\infty$ for all $i$ and $t$. The coefficients $\ppp{b_i}$ are assumed deterministic and unknown. In this case, an appropriate definition of the auxiliary backward channel is
\begin{align}
\tilde{V}\p{\tilde{\bx}\vert\tilde{\by},\btheta,k,q} = \prod_{m=1}^n\p{2\pi\sigma_0^2}^{-1/2}\exp\ppp{-\frac{1}{2\sigma_0^2}\abs{\tilde{x}_m-\tilde{y}_m\sum_{l=0}^k\alpha_le^{\frac{2\pi jlm}{n}}-\sum_{i=1}^q\beta_i\tilde{\phi}_{i,m}}^2}
\end{align}
where now $\btheta\define\p{\sigma_0^2,\alpha_0,\ldots,\alpha_k,\beta_1,\ldots,\beta_q}$ is the parameter vector of the backward channel, $\ppp{\tilde{\phi}_{i,m}}$ is the frequency transformed representation of $\ppp{\phi_{i,t}}$, and $q=q_n$ is assumed to be a monotonically non-decreasing integer-valued sequence such that $q_n\to\infty$ and $q_n = o\p{n^{1/3}}$. Accordingly, the decoding rule will select a message $\bx^i$ that maximizes the metric \eqref{risk} (where $V$ in \eqref{risk} is replaced with $\tilde{V}$), among all $M$ codewords. For simplicity of the exposition and to facilitate the reading of the proof of Theorem \ref{th:1}, we will assume the original model \eqref{GenModel}.
\end{remark}

\section{Proof Outline}\label{sec:proofout}
In this section, before getting deep into the proof of Theorem \ref{th:1}, we discuss the techniques and the main steps which will be used in Section \ref{sec:proof}. In order to facilitate the explanations, we will need the following definitions: Let $\bx$ and $\by$ be arbitrary vectors in $\mathbb{R}^n$ and define
\begin{align}
&\calS_o\p{\bx,\by}\define\ppp{\bx':W\p{\by\vert\bx'}>W\p{\by\vert\bx}},\label{1o}\\
&\calS_u\p{\bx,\by}\define\ppp{\bx':u\p{\bx',\by}>u\p{\bx,\by}}\label{2o},
\end{align}
and 
\begin{align}
\calS_0^\delta\p{\bx,\by}\define\ppp{\bx':\frac{1}{n}\log W\p{\by\vert\bx'}>\frac{1}{n}\log W\p{\by\vert\bx}-\delta},\label{3o}
\end{align}
where $W\p{\by\vert\bx}$ is the conditional pdf associated with the channel. In words, $\calS_o\p{\bx,\by}$ and $\calS_u\p{\bx,\by}$ are simply the sets of prospective incorrect codewords corresponding to the ML decoder, and the proposed universal decoder, respectively, assuming that $\bx$ is the transmitted codewords and that $\by$ is the received vector. The set $\calS_0^\delta\p{\bx,\by}$ is just a $\delta$-perturbed version of $\calS_0\p{\bx,\by}$ which will be used for technical reasons. Finally, we let $\bar{P}_{e,o}\p{R,n}$, $\bar{P}_{e,u}\p{R,n}$, and $\bar{P}_{e,o}^\delta\p{R,n}$ be the average error probabilities associated with the ML decoder, the proposed decoder, and the $\delta$-perturbed decoder (see, \eqref{per1}-\eqref{per3}). 

Generally speaking, the root of our analysis is Lemma \ref{lem:1}, which was asserted and proved in \cite[Lemma 1]{NeriUni}, and can be thought as a continuous extension of \cite[Corollary 1]{ZivUni}. This result relates between $\bar{P}_{e,o}^\delta\p{R,n}$ and $\bar{P}_{e,u}\p{R,n}$ as follows
\begin{align}
\bar{P}_{e,u}\p{R,n}\leq 2\bar{P}_{e,o}^\delta\p{R,n}\pp{\frac{3}{2}+\sup_{\p{\bxt,\byt}\in H_n}\frac{\int_{\calS_u\p{\bxt,\byt}}\mu\p{\bx'}\mathrm{d}\bx'}{\int_{\calS_o^\delta\p{\bxt,\byt}}\mu\p{\bx'}\mathrm{d}\bx'}}.
\label{unRes}
\end{align}
where $\ppp{H_n}_{n\geq1}$ is a sequence of sets of pairs $\p{\bx,\by}$ such that
\begin{align}
\limsup_{n\to\infty}\frac{1}{n}\log\pr\ppp{H_n^c}<-E\p{R}.
\label{priHn}
\end{align}
Whence, we see that in order to show that $\bar{P}_{e,u}\p{R,n}$ and $\bar{P}_{e,o}\p{R,n}$ are exponentially the same, we just need to define a sequence $\ppp{H_n}_{n\geq1}$ such that the ratio in \eqref{unRes}
\begin{align}
\frac{\int_{\calS_u\p{\bxt,\byt}}\mu\p{\bx'}\mathrm{d}\bx'}{\int_{\calS_o^\delta\p{\bxt,\byt}}\mu\p{\bx'}\mathrm{d}\bx'}
\label{subexponentialdec}
\end{align}
is uniformly overbounded by a subexponential function of $n$, i.e., $e^{n\epsilon_n}$ where $\epsilon_n\to0$ as $n\to\infty$ uniformly for all $\p{\bx,\by}\in H_n$. Once this accomplished, the proof of the theorem will be complete. The main question is now how to define the sequence $\ppp{H_n}_{n\geq1}$ properly? To answer this question, let us interpret its role. The set $H_n$ simply divides the space of pairs $\p{\bx,\by}$ into two parts, where in the first part, the supremum in \eqref{unRes} is uniformly bounded by a subexponential function of $n$, and the second part possesses a probability smaller than the desired exponential function $e^{-nE(R)}$ and hence negligible (see, \eqref{priHn}). Obviously, given these requirements one can propose several candidates for $H_n$, namely, the choice is not unique. However, another important property that $\ppp{H_n}_{n\geq1}$ should account for is that the function $n^{-1}\log V\p{\tilde{\bx}\vert\tilde{\by},\btheta,k}$ will be uniformly continuous w.r.t. small perturbations of the sufficient statistics (this idea will be emphasized in the analysis). To summarize, the first part in the forthcoming analysis is to define the sequence $\ppp{H_n}_{n\geq1}$ such that \eqref{priHn} holds true, and that hopefully \eqref{unRes} will hold too. The proposed $\ppp{H_n}_{n\geq1}$ is given in Lemma \ref{lem:Suff}, and the main tool that is used in the proof is large deviations theory. 

Following the first part, in the second part, we will eventually show that the chosen $H_n$ fulfills the desired subexponential behavior of \eqref{subexponentialdec}. Accordingly, we will overbound \eqref{subexponentialdec} within $H_n$ as follows: we will derive an upper bound on the numerator of \eqref{subexponentialdec} and a lower bound on its denominator, and show that these are exponentially equivalent. To this end, we will need to define a conditional typical set of our continuous-valued input-output sequences, establish some of its properties, and particularly to calculate its volume (Lebesgue measure). This typical set of some sequence $\tilde{\bx}$ given $\tilde{\by}$ will contain all the vectors which, within $\epsilon>0$, have the same sufficient statistics as $\tilde{\bx}$ induced by our backward channel (see \eqref{TypConCon} for a precise definition of this set). Then, we will provide upper and lower bounds (which are exponentially of the same order) on the volume of this typical set. To accomplish this, we will use methods that were previously used in \cite{NeriUni,Wasim,Wasim2}, which are based on large deviations theory and methods that are customary to statistical physics. After that, we will show that for any two vectors $\bu$ and $\bv$ that belong to this typical set, the conditional pdf's $W\p{\by\vert\bu}$ and $W\p{\by\vert\bv}$ are exponentially equivalent, that is, for sufficiently large $n$,
\begin{align}
\abs{\frac{1}{n}\log W\p{\by\vert\bu}-\frac{1}{n}\log W\p{\by\vert\bv}}<\zeta
\end{align}
for any $\zeta>0$. Thus, given this property, we can easily provide a lower bound on the denominator of \eqref{subexponentialdec}. Indeed, since $\bx\in\calS_o^\delta\p{\bx,\by}$, then in view of the last result, there exists a sufficiently small $\epsilon>0$ such that the predefined typical set is essentially a subset of $\calS_o^\delta\p{\bxt,\byt}$. Therefore, the integral over $\calS_o^\delta\p{\bxt,\byt}$, in the denominator, can be underestimated as an integral over the typical set, and since we know its volume (or, more precisely, a lower bound on it which is exponentially tight), it is not difficult to provide a lower bound on this integral (see \eqref{eq:1} for more details). Providing an upper bound on the numerator is a little more involved. The underlying idea is to partition the set $\calS_u\p{\bxt,\byt}$ into a subexponential number of conditional types, where for each conditional type, the integral over the respective conditional type is overestimated using the upper bound on the volume. Finally, it will be shown that these two bounds are exponentially equivalent, which implies that \eqref{subexponentialdec} is subexponential function of $n$, as required.  

\section{Proof of Theorem \ref{th:1}}\label{sec:proof}
For completeness, in this section, we will provide again some definitions that were already presented in short in the previous section. Let $\bx$ and $\by$ be arbitrary vectors in $\mathbb{R}^n$ and define $\calS_o\p{\bx,\by}$ and $\calS_u\p{\bx,\by}$ as in eqs. \eqref{1o} and \eqref{2o}, respectively.  
The average error probabilities associated with the ML decoder and the proposed decoder are given by (see, for example, \cite{NeriUni})
\begin{align}
\bar{P}_{e,o}\p{R,n} = 1-\bE\ppp{\pp{1-\int_{\calS_o\p{\bXt,\bYt}}\mu\p{\bx'}\mathrm{d}\bx'}^{2^{nR}-1}}
\label{per1}
\end{align}
and
\begin{align}
\bar{P}_{e,u}\p{R,n} = 1-\bE\ppp{\pp{1-\int_{\calS_u\p{\bXt,\bYt}}\mu\p{\bx'}\mathrm{d}\bx'}^{2^{nR}-1}},
\label{per2}
\end{align}
respectively, where the expectations are taken with respect to (w.r.t.) the joint distribution $\mu\p{\bx}W\p{\by\vert\bx}$, and we use the usual conventions where random vectors are denoted by capital letters in bold face font, and their sample values are denoted by the respective lower case letters. Similar convention will apply to scalar random variables (RVs), which will be denoted with same symbols without the bold face font. Finally, for $\delta>0$ we define the set
\begin{align}
\calS_0^\delta\p{\bx,\by}\define\ppp{\bx':\frac{1}{n}\log W\p{\by\vert\bx'}>\frac{1}{n}\log W\p{\by\vert\bx}-\delta},
\end{align}
and accordingly
\begin{align}
\bar{P}_{e,o}^\delta\p{R,n} = 1-\bE\ppp{\pp{1-\int_{\calS_0^\delta\p{\bXt,\bYt}}\mu\p{\bx'}\mathrm{d}\bx'}^{2^{nR}-1}}.
\label{per3}
\end{align}
Finally, with a slight abuse of notation, we also use the notation ${\calS}_o\p{\tilde{{\bx}},\tilde{\by}}$ which is defined as follows: Let $\tilde{\bx}$ and $\tilde{\by}$ be the Fourier transforms of $\bx$ and $\by$, respectively. Then, ${\calS}_o\p{\tilde{{\bx}},\tilde{\by}}\define\ppp{\tilde{\bx}' = \bF^H\bx':\ \bx'\in\calS_0\p{\bx,\bF^H\tilde{\by}}}$ where $\bF$ is the DFT matrix, namely, $\bF = \ppp{e^{j2\pi ml/n}/\sqrt{n}}_{m,l=0}^{n-1}$. 

As was discussed earlier, our goal is to compare the exponential behavior of $\bar{P}_{e,u}\p{R,n}$ to that of $\bar{P}_{e,o}\p{R,n}$. To this end, we will instead compare the exponential behavior of $\bar{P}_{e,u}\p{R,n}$ to that of $\bar{P}_{e,o}^\delta\p{R,n}$ for small $\delta>0$. In the final step of the proof, this will be justified by showing that
\begin{align}
\limsup_{n\to\infty}\frac{1}{n}\pp{\log\bar{P}_{e,o}^\delta\p{R,n} - \log\bar{P}_{e,o}\p{R,n}}\leq\delta'
\label{justified}
\end{align}
where $\delta'\to0$ as $\delta\to0$ and $\Delta\to0$. In the analysis, we will use the following lemma \cite[Lemma 1 pp. 1263]{NeriUni}.
\begin{lemma}\label{lem:1}
Let $\ppp{H_n}_{n\geq1}$ be a sequence of sets of pairs $\p{\tilde{\bx},\tilde{\by}}$ of $n$-dimensional vectors such that
\begin{align}
\limsup_{n\to\infty}\frac{1}{n}\log\pr\ppp{H_n^c}<-E\p{R}
\label{Hconstranit}
\end{align}
Then, for all large $n$,
\begin{align}
\bar{P}_{e,u}\p{R,n}\leq 2\bar{P}_{e,o}^\delta\p{R,n}\pp{\frac{3}{2}+\sup_{\p{\tilde{\bxt},\tilde{\byt}}\in H_n}\frac{\int_{\calS_u\p{\tilde{\bxt},\tilde{\byt}}}\mu\p{\bx'}\mathrm{d}\bx'}{\int_{\calS_o^\delta\p{\tilde{\bxt},\tilde{\byt}}}\mu\p{\bx'}\mathrm{d}\bx'}}.
\label{supremum}
\end{align}
\end{lemma}

Thus, by using Lemma \ref{lem:1}, we see that in order to show that $\bar{P}_{e,u}\p{R,n}$ and $\bar{P}_{e,o}\p{R,n}$ are exponentially the same, we just need to find a sequence $\ppp{H_n}_{n\geq1}$ such that the ratio
\begin{align}
\frac{\int_{\calS_u\p{\tilde{\bxt},\tilde{\byt}}}\mu\p{\bx'}\mathrm{d}\bx'}{\int_{\calS_o^\delta\p{\tilde{\bxt},\tilde{\byt}}}\mu\p{\bx'}\mathrm{d}\bx'}
\label{overBound}
\end{align}
is uniformly overbounded by a subexponential function of $n$, i.e., $e^{n\epsilon_n}$ where $\epsilon_n\to0$ as $n\to\infty$ uniformly for all $\p{\tilde{\bx},\tilde{\bx}}\in H_n$. 
For a given pair $\p{\tilde{\bx},\tilde{\by}}$, let us define $\hat{\btheta} = \p{\hat{\sigma}_0^2,\hat{\alpha}_0,\ldots,\hat{\alpha}_k}$ to be
\begin{align}
\hat{\btheta}\define\arg\max_{\theta}V\p{\tilde{\bx}\vert\tilde{\by},\btheta,k}.
\end{align}
The set $H_n$ will be parametrized by a parameter $B>0$ and defined as follows
\begin{align}
H_n\p{B}\define \ppp{\p{\tilde{\bx},\tilde{\by}}:\;\abs{\frac{1}{n}\sum_{m=0}^{n-1}\abs{\tilde{x}_m}^2-\sigma_x^2}\leq\Delta\sigma_x^2,\;\frac{1}{n}\sum_{m=0}^{n-1}\abs{\tilde{y}_m}^2\leq B,\;\hat{\sigma}_0^2\geq\frac{1}{B}}.
\label{HnSet}
\end{align}
We have the following result. 
\begin{lemma}\label{lem:Suff}
There exists a sufficiently large $B$ such that $\ppp{H_n\p{B}}_{n\geq1}$ satisfies \eqref{Hconstranit}. 
\end{lemma}

\begin{proof}[Proof of Lemma \ref{lem:Suff}]
By the union bound we have that
\begin{align}
\pr\ppp{H_n^c\p{B}}\leq \pr\ppp{\frac{1}{n}\sum_{t=0}^{n-1}Y_t^2> B} +\pr\ppp{\hat{\sigma}_0^2< B^{-1}}.
\label{Lins}
\end{align}
Thus, it should be shown that if $B$ is sufficiently large, both probabilities on the right-hand side of \eqref{Lins} decays faster than $e^{-nE\p{R}}$. Regarding the first term, note that
\begin{align}
\frac{1}{n}\sum_{t=0}^{n-1}y_t^2 &\leq \pp{\sqrt{\frac{1}{n}\norm{\bH\bx}^2}+\sqrt{\frac{1}{n}\sum_{t=0}^{n-1}w_t^2}}^2\\
&\leq\pp{\sqrt{\frac{1}{n}\norm{\bx}^2}\sqrt{\norm{\bH^T\bH}_s}+\sqrt{\frac{1}{n}\sum_{t=0}^{n-1}w_t^2}}^2\\
&\leq\pp{\sqrt{\sigma_x^2\p{1+\Delta}}\norm{\bH}_s+\sqrt{\frac{1}{n}\sum_{t=0}^{n-1}w_t^2}}^2
\end{align}
where $\norm{\cdot}_s$ denotes the spectral norm, and in the second inequality we have used the fact that $\abs{\tr{\bA\bB}}\leq\norm{\bB}_s\tr\p{\bA}$ for any $\bB$ and nonnegative definite matrix $\bA$. Due to the fact that $\ppp{h_m}\in\ell_1$ (essentially, $\ppp{h_m}\in\ell_2$ is suffice here) it can be shown that \cite{Gray} the spectral norm $\norm{\bH}_s$ is uniformly bounded, that is for all matrix dimension $n$ we have that $\norm{\bH}_s\leq M$ where $M>0$. Therefore, we obtain that
\begin{align}
\pr\ppp{\frac{1}{n}\sum_{t=0}^{n-1}Y_t^2> B} \leq \pr\ppp{\frac{1}{n}\sum_{t=0}^{n-1}W_t^2> \p{\sqrt{B}-M\sqrt{\sigma_x^2\p{1+\Delta}}}^2}
\end{align}
which can be made less than $e^{-nE\p{R}}$ by selecting a sufficiently large $B$, as can be shown by a simple application of the Chernoff bound. As for the remaining terms: by taking the gradient of $V\p{\tilde{\bx}\vert\tilde{\by},\btheta,k}$ w.r.t. $\btheta$, we obtain that the components of $\hat{\btheta}$ are given by the solutions of the following set of equations
\begin{align}
\sum_{m=0}^{n-1}\tilde{x}_m\tilde{y}_m^*e^{-\frac{2\pi jmq}{n}} = \sum_{m=0}^{n-1}\abs{\tilde{y}_m}^2e^{-\frac{2\pi jmq}{n}}\sum_{l=0}^k\hat{\alpha}_le^{\frac{2\pi jml}{n}}, \ \ \ \text{for}\ q = 0,\ldots,k,
\label{alphaEst}
\end{align}
and
\begin{align}
\hat{\sigma}_0^2 = \frac{1}{n}\sum_{m=0}^{n-1}\abs{\tilde{x}_m-\tilde{y}_m\sum_{l=1}^k\hat{\alpha}_le^{\frac{2\pi jml}{n}}}^2.
\label{alphaEst2}
\end{align}
Note that
\begin{align}
\pr\ppp{\hat{\sigma}_0^2< B^{-1}}&\leq\pr\ppp{\hat{\sigma}_0^2< B^{-1},\;\frac{1}{n}\sum_{t=0}^{n-1}\abs{\tilde{Y}_m}^2\leq \sqrt{B},\min\limits_{0\leq m\leq n-1}\abs{\tilde{Y}_m}^2\geq\tau} \nonumber\\
&\ \ + \pr\ppp{\frac{1}{n}\sum_{t=0}^{n-1}\abs{\tilde{Y}_m}^2>\sqrt{B}}+\pr\ppp{\max\limits_{0\leq m\leq n-1}\abs{\tilde{Y}_m}^2\leq\tau} \nonumber\\
&\leq\pr\ppp{\hat{\sigma}_0^2< B^{-1},\;\frac{1}{n}\sum_{t=0}^{n-1}\abs{\tilde{Y}_m}^2\leq \sqrt{B},\min\limits_{0\leq m\leq n-1}\abs{\tilde{Y}_m}^2\geq\tau} \nonumber\\
&\ \ + \pr\ppp{\frac{1}{n}\sum_{t=0}^{n-1}\abs{\tilde{Y}_m}^2>\sqrt{B}}+\pr\ppp{\frac{1}{n}\sum_{t=0}^{n-1}\abs{\tilde{Y}_m}^2\leq\tau}
\label{Lse2}
\end{align}
where $\tau>0$. As before, the exponential decay rate of the last two terms on the right-hand side of \eqref{Lse2} can be made arbitrarily large by selecting a sufficiently large $B$ and sufficiently small $\tau$. As for the first term, we first note that by using \eqref{alphaEst}, we have
\begin{align}
\re\ppp{\sum_{m=0}^{n-1}\sum_{q=0}^k\tilde{x}_m\tilde{y}_m^*\alpha_q^*e^{-\frac{2\pi jmq}{n}}} &= \re\ppp{\sum_{m=0}^{n-1}\sum_{q=0}^k\alpha_q^*\abs{\tilde{y}_m}^2e^{-\frac{2\pi jmq}{n}}\sum_{l=0}^k\hat{\alpha}_le^{\frac{2\pi jml}{n}}}\\
& =  \sum_{m=0}^{n-1}\abs{\tilde{y}_m}^2\abs{\sum_{l=0}^k\hat{\alpha}_le^{\frac{2\pi jml}{n}}}^2.
\label{alphaEst3}
\end{align}
Thus, using the last result we obtain
\begin{align}
\hat{\sigma}_0^2 &= \frac{1}{n}\sum_{m=0}^{n-1}\abs{\tilde{x}_m-\tilde{y}_m\sum_{l=0}^k\hat{\alpha}_le^{\frac{2\pi jml}{n}}}^2\label{quadNorm}\\
&= \frac{1}{n}\sum_{m=0}^{n-1}\abs{\tilde{x}_m}^2-2\re\ppp{\sum_{m=0}^{n-1}\sum_{l=0}^k\tilde{x}_m\tilde{y}_m^*\alpha_l^*e^{-\frac{2\pi jml}{n}}}+ \frac{1}{n}\sum_{m=0}^{n-1}\abs{\tilde{y}_m}^2\abs{\sum_{l=0}^k\hat{\alpha}_le^{\frac{2\pi jml}{n}}}^2\\
&=\frac{1}{n}\sum_{m=0}^{n-1}\abs{\tilde{x}_m}^2- \frac{1}{n}\sum_{m=0}^{n-1}\abs{\tilde{y}_m}^2\abs{\sum_{l=0}^k\hat{\alpha}_le^{\frac{2\pi jml}{n}}}^2,
\label{lowerVar}
\end{align}
which in turn must be nonnegative, and hence
\begin{align}
\frac{1}{n}\sum_{m=0}^{n-1}\abs{\tilde{y}_m}^2\abs{\sum_{l=0}^k\hat{\alpha}_le^{\frac{2\pi jml}{n}}}^2&\leq\frac{1}{n}\sum_{m=0}^{n-1}\abs{\tilde{x}_m}^2\leq\sigma^2_x\p{1+\Delta}.
\label{upperOne}
\end{align}
Thus, given that $\min\limits_{0\leq m\leq n-1}\abs{\tilde{y}_m}^2\geq\tau$, by using \eqref{lowerVar} we obtain that
\begin{align}
\frac{1}{n}\sum_{m=0}^{n-1}\abs{\tilde{y}_m}^2\abs{\sum_{l=0}^k\hat{\alpha}_le^{\frac{2\pi jml}{n}}}^2&\geq\tau\frac{1}{n}\sum_{m=0}^{n-1}\abs{\sum_{l=0}^k\hat{\alpha}_le^{\frac{2\pi jml}{n}}}^2\\
&=\tau\sum_{l=0}^k\sum_{r=0}^k\hat{\alpha}_l\hat{\alpha}_r^*\frac{1}{n}\sum_{m=0}^{n-1}e^{\frac{2\pi jm(l-r)}{n}}\\
&=\tau\sum_{l=0}^k\abs{\hat{\alpha}_l}^2. 
\end{align}
Therefore, invoking \eqref{upperOne}, we finally obtain that
\begin{align}
\sum_{l=0}^k\abs{\hat{\alpha}_l}^2\leq \frac{\sigma_x^2\p{1+\Delta}}{\tau}\define C\p{\tau,\Delta}.
\label{upperBound}
\end{align}
Now, recall that $\ppp{\hat{\alpha}_l}$ minimizes the quadratic norm
$$
\frac{1}{n}\sum_{m=0}^{n-1}\abs{\tilde{x}_m-\tilde{y}_m\sum_{l=0}^k\alpha_le^{\frac{2\pi jml}{n}}}^2
$$
over all vectors $\balpha = \p{\alpha_0,\ldots,\alpha_k}$ in $\mathbb{C}^{k+1}$. Also, due to \eqref{upperBound}, the minimizing vector must lie in the $\p{k+1}$-dimensional hypersphere $\balpha^H\balpha\leq C\p{\tau,\Delta}$. Now, fix $\delta>0$ and define the grid $\calG\define \ppp{\delta\cdot i:\;i = -\left\lceil C\p{\tau,\Delta}/\delta\right\rceil,\ldots,-1,0,1,\ldots,\left\lceil C\p{\tau,\Delta}/\delta\right\rceil}$, and let $\calG^{k+1}$ designate the $\p{k+1}$th Cartesian power of $\calG$. From the uniform continuity of the above quadratic form within the set of all energy limited vectors $\by$, one can find a sufficiently small value of $\delta$ (depending on $C$) such that there exists a vector $\balpha = \balpha_R+j\balpha_I$ where $\balpha_R,\balpha_I\in\calG^{k+1}$, i.e., the nearest neighbor of the minimizer, satisfying (given of course the event that $\hat{\sigma}_0^2<B^{-1}$)
\begin{align}
\frac{1}{n}\sum_{m=0}^{n-1}\abs{\tilde{x}_m-\tilde{y}_m\sum_{l=0}^k\alpha_le^{\frac{2\pi jml}{n}}}^2\leq\frac{2}{B}+\delta'
\end{align}
where $\delta'$ is a sufficiently small value (depending on $\delta$). For brevity, in the following, we will omit this negligible additive term. Whence
\begin{align}
&\pr\ppp{\hat{\sigma}_0^2< B^{-1},\;\frac{1}{n}\sum_{t=0}^{n-1}\abs{\tilde{Y}_m}^2\leq \sqrt{B},\min\limits_{0\leq m\leq n-1}\abs{\tilde{Y}_m}^2\geq\tau} \nonumber\\
&= \pr\ppp{\frac{1}{n}\sum_{m=0}^{n-1}\abs{\tilde{X}_m-\tilde{Y}_m\sum_{l=0}^k\hat{\alpha}_le^{\frac{2\pi jml}{n}}}^2< \frac{1}{B},\;\frac{1}{n}\sum_{t=0}^{n-1}\abs{\tilde{Y}_m}^2\leq \sqrt{B},\min\limits_{0\leq m\leq n-1}\abs{\tilde{Y}_m}^2\geq\tau}\\
&\leq \pr\ppp{\bigcup_{\balphat_R,\balphat_I\in\calG^{k+1}}\ppp{\frac{1}{n}\sum_{m=0}^{n-1}\abs{\tilde{X}_m-\tilde{Y}_m\sum_{l=0}^k\alpha_le^{\frac{2\pi jml}{n}}}^2< \frac{2}{B},\;\frac{1}{n}\sum_{t=0}^{n-1}\abs{\tilde{Y}_m}^2\leq \sqrt{B}}}\\
&\leq\sum_{\balphat_R,\balphat_I\in\calG^{k+1}}\pr\ppp{\frac{1}{n}\sum_{m=0}^{n-1}\abs{\tilde{X}_m-\tilde{Y}_m\sum_{l=0}^k\alpha_le^{\frac{2\pi jml}{n}}}^2< \frac{2}{B},\;\frac{1}{n}\sum_{t=0}^{n-1}\abs{\tilde{Y}_m}^2\leq \sqrt{B}}\\
&\leq \p{\left\lceil \frac{C\p{\tau,\Delta}}{\delta}\right\rceil}^{2k+2}\cdot\max_{\balphat_R,\balphat_I\in\calG^{k+1}}\;\pr\ppp{\frac{1}{n}\sum_{m=0}^{n-1}\abs{\tilde{X}_m-\tilde{Y}_m\sum_{l=0}^k\alpha_le^{\frac{2\pi jml}{n}}}^2< \frac{2}{B},\;\frac{1}{n}\sum_{t=0}^{n-1}\abs{\tilde{Y}_m}^2\leq \sqrt{B}}.
\end{align}
Let us show that the term on the right-most side of can be made exponentially less than $e^{-nE(R)}$. Define the set
\begin{align}
\calF_{\alpha}\define\ppp{\p{\tilde{\bx},\tilde{\by}}:\;\frac{1}{n}\sum_{m=0}^{n-1}\abs{\tilde{x}_m-\tilde{y}_m\sum_{l=0}^k\alpha_le^{\frac{2\pi jml}{n}}}^2< \frac{2}{B},\;\frac{1}{n}\sum_{m=0}^{n-1}\abs{\tilde{y}_m}^2\leq \sqrt{B}}.
\end{align}
Accordingly, define an auxiliary joint density
\begin{align}
g\p{\tilde{\bx},\tilde{\by}} = \frac{1}{\p{2\pi^2/\sqrt{B}}^{n}}\prod_{m=0}^{n-1}\exp\ppp{-\frac{B}{2}\abs{\tilde{x}_m-\tilde{y}_m\sum_{l=0}^k\alpha_le^{\frac{2\pi jml}{n}}}^2}\exp\ppp{-\frac{1}{\sqrt{B}}\abs{\tilde{y}_m}^2}.
\end{align}
Thus,
\begin{align}
1&\geq\int_{\calF_{\alpha}}g\p{\mathrm{d}\bx,\mathrm{d}\by}\\
&\geq\frac{\text{Vol}\ppp{\calF_{\alpha}}}{\p{2\pi^2/\sqrt{B}}^{n}}\inf_{\p{\bxt,\byt}\in\calF_{\alpha}}\ppp{\prod_{m=0}^{n-1}\exp\ppp{-\frac{B}{2}\abs{\tilde{x}_m-\tilde{y}_m\sum_{l=0}^k\alpha_le^{\frac{2\pi jml}{n}}}^2}\exp\ppp{-\frac{1}{\sqrt{B}}\abs{\tilde{y}_m}^2}}\\
&\geq\frac{\text{Vol}\ppp{\calF_{\alpha}}}{\p{2\pi^2/\sqrt{B}}^{n}}\exp\ppp{-2n}\\
& = \text{Vol}\ppp{\calF_{\alpha}}\p{\frac{2\pi^2e^2}{\sqrt{B}}}^{-n},
\end{align}
and therefore $\text{Vol}\ppp{\calF_{\alpha}}\leq\p{2\pi^2e^2/\sqrt{B}}^n$. Thus, we now obtain that
\begin{align}
\pr\ppp{\calF_{\alpha}} &= \int_{\p{\bxt,\byt}\in\calF_{\alpha}}\mu\p{\bx}W\p{\by\vert\bx}\mathrm{d}\bx\mathrm{d}\by\\
&\leq \text{Vol}\ppp{\calF_{\xi}}\p{2\pi\sigma^2}^{-n/2}\nu^{-1}\\
& \leq \p{2\pi\sigma^2}^{-n/2}\nu^{-1}\exp\ppp{-\frac{n}{2}\log\p{\frac{B}{2\pi^2e^2}}}\\
& = \nu^{-1}\exp\ppp{-\frac{n}{2}\log\p{\frac{B\sigma^2}{\pi e^2}}}
\end{align}
which, again, can be made less than $e^{-nE\p{R}}$ by selecting $B$ sufficiently large. 
\end{proof}

To overbound \eqref{overBound} within $H_n\p{B}$, we derive an upper bound on its numerator and a lower bound on its denominator, and show that these are exponentially equivalent. To this end, we first need to define a conditional typical set of our continuous-valued input-output sequences and establish some of its properties. For a given pair of vectors $\p{\tilde{\bx},\tilde{\by}}$ and $\epsilon > 0$, define the $k$th order conditional $\epsilon$-type of $\tilde{\bx}$ given $\tilde{\by}$ as
\begin{align}
\calT_\epsilon^k\p{\tilde{\bx}\vert\tilde{\by}}\define&\left\{\tilde{\bx}'\in\mathbb{C}^n:\;\abs{\frac{1}{n}\sum_{m=0}^{n-1}\abs{\tilde{x}_m}^2-\frac{1}{n}\sum_{m=0}^{n-1}\abs{\tilde{x}_m'}^2}\leq\epsilon\right.\nonumber\\
&\left. \ \ \ \ \abs{\frac{1}{n}\sum_{m=0}^{n-1}\re\ppp{\tilde{x}_m\tilde{y}_m^*e^{-\frac{2\pi jlm}{n}}}-\frac{1}{n}\sum_{m=0}^{n-1}\re\ppp{\tilde{x}'_m\tilde{y}_m^*e^{-\frac{2\pi jlm}{n}}}}\leq\epsilon,\;l=0,\ldots,k,\right.\nonumber\\
&\left. \ \ \ \ \abs{\frac{1}{n}\sum_{m=0}^{n-1}\Img\ppp{\tilde{x}_m\tilde{y}_m^*e^{-\frac{2\pi jlm}{n}}}-\frac{1}{n}\sum_{m=0}^{n-1}\Img\ppp{\tilde{x}'_m\tilde{y}_m^*e^{-\frac{2\pi jlm}{n}}}}\leq\epsilon,\;l=0,\ldots,k\right\}.
\label{TypConCon}
\end{align}
This set is regarded as a conditional type of $\tilde{\bx}$ given $\tilde{\by}$ as it contains all vectors which, within $\epsilon$, have the same sufficient statistics as $\tilde{\bx}$ induced by our backward channel. In the following, we will show that for every conditional type $\calT_\epsilon^k\p{\tilde{\bx}\vert\tilde{\by}}$, and for any two vectors $\bu$ and $\bv$ in $\calT_\epsilon^k\p{\tilde{\bx}\vert\tilde{\by}}$, the conditional pdf's $W\p{{\by}\vert\bu}$ and $W\p{{\by}\vert\bv}$ are exponentially equivalent. This property will be used later on. To show that this is indeed the case, we will need the following lemma. 
\begin{lemma}\label{lem:3}
Let $L$ and $n_b$ be natural numbers such that\footnote{Without loss of generality, it is assumed that $n_b$ (\emph{bin} size) is a divisor of $n$, and that all $L$ bins have the same size.} $L=n/n_b$. Define the sets $\calG_{1,\epsilon} \define \ppp{\epsilon\cdot i:\;i=0,1,\ldots,\left\lceil LP_x/\epsilon\right\rceil}$. Also, let
\begin{align}
\hat{\calT}_\epsilon^k\p{\tilde{\bx}\vert\tilde{\by}} \define \ppp{\bigcup\limits_{\substack{\boldsymbol{\mathcal{P}}^{\epsilon}}}\ \bigtimes_{l=1}^L\mathscr{B}_{l}^{\epsilon}\p{P_l}}\bigcap\tilde{\calT}_\epsilon\p{\tilde{\bx}\vert\tilde{\by}}
\end{align}
where $\bigtimes$ designates a Cartesian product, and
\begin{align}
\tilde{\calT}_\epsilon^k\p{\tilde{\bx}\vert\tilde{\by}}\define&\left\{\tilde{\bx}'\in\mathbb{C}^n:\;\abs{\frac{1}{n}\sum_{m=0}^{n-1}\re\ppp{\tilde{x}_m\tilde{y}_m^*e^{-\frac{2\pi jlm}{n}}}-\frac{1}{n}\sum_{m=0}^{n-1}\re\ppp{\tilde{x}'_m\tilde{y}_m^*e^{-\frac{2\pi jlm}{n}}}}\leq\epsilon,\;l=0,\ldots,k,\right.\nonumber\\
&\left. \ \ \ \abs{\frac{1}{n}\sum_{m=0}^{n-1}\Img\ppp{\tilde{x}_m\tilde{y}_m^*e^{-\frac{2\pi jlm}{n}}}-\frac{1}{n}\sum_{m=0}^{n-1}\Img\ppp{\tilde{x}'_m\tilde{y}_m^*e^{-\frac{2\pi jlm}{n}}}}\leq\epsilon,\;l=0,\ldots,k\right\},
\end{align}
and
\begin{align}
\mathscr{B}_{l}^\epsilon\p{P_l} &\define \left\{\tilde{\bx}' \in\mathbb{C}^{n_b}:\;\abs{\norm{\tilde{\bx}'}^2-n_{b}P_l}\leq\epsilon,\nonumber\right\}
\end{align}
where 
\begin{align}
\label{powerSetg}
\boldsymbol{\mathcal{P}}^{\epsilon}&\define\ppp{\bP\in\calG_{1,\epsilon}^L:\;\abs{\frac{1}{L}\sum_{i=1}^LP_i-\frac{1}{n}\sum_{m=0}^{n-1}\abs{\tilde{x}_m}^2} \leq\epsilon}
\end{align}
where $\calG_{1,\epsilon}^L$ is the $L$th Cartesian power of $\calG_{1,\epsilon}$. Then,
\begin{align}
\calT_\epsilon^k\p{\tilde{\bx}\vert\tilde{\by}}\subseteq\hat{\calT}_\epsilon^k\p{\tilde{\bx}\vert\tilde{\by}}.
\end{align}
\end{lemma}
\begin{proof}
See Appendix \ref{app:1}.
\end{proof}

Intuitively speaking, the difference between $\hat{\calT}_\epsilon^k\p{\tilde{\bx}\vert\tilde{\by}}$ and $\calT_\epsilon^k\p{\tilde{\bx}\vert\tilde{\by}}$ is that in the former we split each sequence $\tilde{\bx}$ into $L$ bins, where in each bin we fix the energy. Indeed, let $\tilde{\bu},\tilde{\bv}\in{\calT}_\epsilon^k\p{\tilde{\bx}\vert\tilde{\by}}$. Due to Lemma \ref{lem:3}, we also have that $\tilde{\bu},\tilde{\bv}\in\hat{\calT}_\epsilon^k\p{\tilde{\bx}\vert\tilde{\by}}$. Then,
\begin{align}
\abs{\frac{1}{n}\log W\p{\by\vert\bu} - \frac{1}{n}\log W\p{\by\vert\bv}} = \frac{1}{2\sigma^2}\abs{\frac{1}{n}\sum_{t=0}^{n-1}\p{y_t-\sum_{l=0}^{k}h_lu_{t-l}}^2 - \frac{1}{n}\sum_{t=0}^{n-1}\p{y_t-\sum_{l=0}^{k}h_lv_{t-l}}^2}.
\end{align}
Recall that the model in \eqref{GenModel} can be represented in the following vector form $\by = \bA\bx+\bw$ where $\bA$ is a Toeplitz matrix formed by the generating sequence $\ppp{h_l}$, that is $\bA = \ppp{a_{i,j}}_{i,j} = \ppp{h_{i-j}}_{i,j}$. Now, by using the spectral decomposition theorem \cite{Widom}, we know that there exists a unitary matrix $\bR$ that diagonalizes $\bA$. Accordingly, let $\ppp{\lambda_l}_{l=1}^n$ denote the singular values associated with this transformation. Thus, we obtain that
\begin{align}
\abs{\frac{1}{n}\log W\p{\by\vert\bu} - \frac{1}{n}\log W\p{\by\vert\bv}} = \frac{1}{2\sigma^2}\abs{\frac{1}{n}\sum_{m=0}^{n-1}\abs{\hat{y}_m-\lambda_m\hat{u}_m}^2 - \frac{1}{n}\sum_{t=0}^{n-1}\abs{\hat{y}_m-\lambda_m\hat{v}_m}^2}
\end{align}
where $\hat{\by} = \bR\by$ and similarly for $\ppp{\hat{u}_m}$ and $\ppp{\hat{v}_m}$. Continuing, we see that
\begin{align}
\abs{\frac{1}{n}\log W\p{\by\vert\bu} - \frac{1}{n}\log W\p{\by\vert\bv}} \leq &\frac{1}{\sigma^2}\abs{\frac{1}{n}\sum_{m=0}^{n-1}\re\ppp{\hat{y}_m^*\lambda_m\hat{u}_m} - \frac{1}{n}\sum_{m=0}^{n-1}\re\ppp{\hat{y}_m^*\lambda_m\hat{v}_m}}\nonumber\\
&+\frac{1}{2\sigma^2}\abs{\frac{1}{n}\sum_{m=0}^{n-1}\abs{\lambda_m}^2\p{\abs{\hat{u}_m}^2-\abs{\hat{v}_m}^2}}.
\label{ineqaFirSec}
\end{align}
Now, we note that by Szeg\"{o}'s theorem \cite{Gray,Widom,Szego,Bottcher}, the Fourier basis asymptotically diagonalizes Toeplitz matrices. Accordingly, the asymptotic eigenvalues are given by the DFT of the generating sequence $\ppp{h_l}$, that is, for sufficiently large enough $n$ and any $\varepsilon>0$, we have that \cite{Widom} 
\begin{align}
\abs{\lambda_m-\sum_{l=0}^kh_le^{-2\pi jml/n}}\leq\varepsilon,\ \ m=0,\ldots,n-1,
\label{approxeig}
\end{align} 
and by the same token\footnote{Another approach is to first assume that $\bA$ is a circulant matrix, and then the Fourier basis exactly diagonalizes $\bA$ for any $n$, that is, the eigenvectors are given by the DFT matrix, and the eigenvalues are given by the DFT of $\ppp{h_l}$. Then, when taking the limit $n\to\infty$, using Szeg\"{o}'s theorem, this assumption can be dropped.}, since the Fourier basis asymptotically diagonalizes $\bA$, the eigenvectors matrix $\bR$ asymptotically equal to the Fourier basis $\bF$. Thus, using \eqref{TypConCon} and \eqref{approxeig}, we see that
\begin{align}
&\abs{\frac{1}{n}\sum_{m=0}^{n-1}\re\ppp{\hat{y}_m^*\lambda_m\hat{u}_m} - \frac{1}{n}\sum_{m=0}^{n-1}\re\ppp{\hat{y}_m^*\lambda_m\hat{v}_m}}\nonumber\\
&\ \ \ \leq\sum_{l=0}^k\abs{h_l}\abs{\varepsilon+\frac{1}{n}\sum_{m=0}^{n-1}\re\ppp{\tilde{y}_m^*\tilde{u}_me^{-2\pi jml/n}} - \frac{1}{n}\sum_{m=0}^{n-1}\re\ppp{\tilde{y}_m^*\tilde{v}_me^{-2\pi jml/n}}}\\
&\ \ \ \leq \p{\epsilon+\varepsilon}\sum_{l=0}^k\abs{h_l}\leq \p{\epsilon+\varepsilon} \cdot C_1
\end{align}
where in the last inequality we have used the fact that $\ppp{h_k}$ is absolutely summable. Now, regarding the second term on the right hand side (r.h.s.) of \eqref{ineqaFirSec}, we use the following approximation argument (which is asymptotically tight), that was used in \cite[Sec. VI]{Wasim}. Recall that due to Szeg\"{o}'s theorem, we know that the Fourier basis asymptotically diagonalizes $\bA$, and that there exists a frequency response $H\p{\omega}$ that corresponds to the linear system induced by $\bA$, and is given by the Fourier transform of the sequence $\ppp{h_i}$. Then, we use the fact that every continuous function can be approximated arbitrarily well by a sequence of staircase functions with sufficiently small spacing between jumps. In other words, we approximate the continuous frequency response $H\p{\omega}$ by a staircase function and then we take the width of each stair to zero. This approximation in turn corresponds to assuming that the eigenvalues, $\ppp{\lambda_m}$, are piecewise constant over the various $L$ bins (see Lemma \ref{lem:3}). At the final stage of the analysis (after taking the limit $n\to\infty$), we will take the limit $L\to\infty$ so that this approximation becomes superfluous. Thus, under this approximation, we obtain that 
\begin{align}
\abs{\frac{1}{n}\sum_{m=0}^{n-1}\abs{\lambda_m}^2\p{\abs{\tilde{u}_m}^2-\abs{\tilde{v}_m}^2}} & = \abs{\frac{1}{L}\sum_{l=1}^{L}\frac{1}{n_b}\sum_{m\in\calI_l}\abs{\lambda_m}^2\p{\abs{\tilde{u}_m}^2-\abs{\tilde{v}_m}^2}}\\
& = \abs{\frac{1}{L}\sum_{l=1}^{L}\abs{\lambda_l}^2\frac{1}{n_b}\sum_{m\in\calI_l}\p{\abs{\tilde{u}_m}^2-\abs{\tilde{v}_m}^2}}\\
&\leq \frac{1}{L}\sum_{l=1}^{L}\abs{\lambda_l}^2\abs{\frac{1}{n_b}\sum_{m\in\calI_l}\p{\abs{\tilde{u}_m}^2-\abs{\tilde{v}_m}^2}}\\
&\leq\frac{\epsilon}{L}\sum_{l=1}^{L}\abs{\lambda_l}^2\leq\frac{\epsilon}{L}\sum_{l=1}^{L}\pp{\sum_{v=0}^k\abs{h_v}}^2\leq\epsilon\cdot C_1^2.
\end{align}
Thus, we have shown that 
\begin{align}
\abs{\frac{1}{n}\log W\p{\by\vert\bu} - \frac{1}{n}\log W\p{\by\vert\bv}} \leq \frac{\p{\epsilon+\varepsilon}}{\sigma^2}\cdot C_1\p{1+C_1}.
\label{exponenEquiva}
\end{align}
Clearly, the right-most side of \eqref{exponenEquiva} can be made arbitrarily small by choosing $\epsilon$ sufficiently small and $n,L$ sufficiently large. Similarly, $\mu\p{\bu}$ and $\mu\p{\bv}$ are also exponentially equivalent, provided that they both belong to the support of $\mu\p{\cdot}$, namely, 
\begin{align}
\abs{\frac{1}{n}\log \mu\p{\bu} - \frac{1}{n}\log \mu\p{\bv}} \leq \epsilon\cdot C_2
\end{align}
for some constant $C_2$. Next, we provide upper and lower bounds on the volume of $\calT_\epsilon^k\p{\tilde{\bx}\vert\tilde{\by}}$, where the volume of a set $\calA\subset\mathbb{R}^n$ is defined as $\text{Vol}\ppp{\calA}\define\int_\calA\mathrm{d}\bx$. 
\begin{lemma}\label{lem:Volumes}
Let $\p{\bx,\by}\in H_n\p{B}$ for some $B>0$. Then, for every sufficiently small $\epsilon>0$, the volume of $\calT_\epsilon^k\p{\tilde{\bx}\vert\tilde{\by}}$ is bounded as follows
\begin{align}
\frac{\exp\ppp{-n\epsilon f\p{B,\Delta,k}}}{\max_{\thetavecsc}V\p{\tilde{\bx}\vert\tilde{\by},\btheta,k}}\pp{1-(2k+12)\frac{B^2}{n\epsilon^2}}\leq\text{Vol}\ppp{\calT_\epsilon^k\p{\tilde{\bx}\vert\tilde{\by}}}\leq \frac{\exp\ppp{n\epsilon f\p{B,\Delta,k}}}{\max_{\thetavecsc}V\p{\tilde{\bx}\vert\tilde{\by},\btheta,k}},
\end{align}
in which $f\p{B,\Delta,k}\define B\pp{1+\sqrt{k+1}\cdot C\p{B^{-1},\Delta}}$ where $C\p{\cdot,\cdot}$ is defined in \eqref{upperBound}.
\end{lemma}
\begin{proof}[Proof of Lemma \ref{lem:Volumes}]
Fix a pair $\p{\tilde{\bx},\tilde{\by}}\in H_n\p{B}$ and let 
\begin{align}
&\rho_{xx} \define n^{-1}\sum_{m=0}^{n-1}\abs{\tilde{x}_m}^2,\\
&\rho^l_{R} \define n^{-1}\sum_{m=0}^{n-1}\re\ppp{\tilde{x}_m\tilde{y}_m^*e^{-\frac{2\pi jlm}{n}}}, \ l=0,\ldots,k,
\end{align}
and
\begin{align}
&\rho^{l}_{I} \define n^{-1}\sum_{m=0}^{n-1}\Img\ppp{\tilde{x}_m\tilde{y}_m^*e^{-\frac{2\pi jlm}{n}}}, \ l=0,\ldots,k.
\end{align}
Also, let $\hat{\btheta}$ designate the vector of parameters $\p{\sigma_0^2,\alpha_0,\ldots,\alpha_k}$ that corresponds to the solution of the following set of equations 
\begin{align}
\bE_V\ppp{\sum_{m=0}^{n-1}\abs{\tilde{X}_m}^2} = n\rho_{x,x},
\label{LLN1}
\end{align} 
and
\begin{align}
\bE_V\ppp{\sum_{m=0}^{n-1}\re\ppp{\tilde{X}_m\tilde{y}_m^*e^{-\frac{2\pi jlm}{n}}}} = n\rho^l_{R},\ l=0,\ldots,k,
\label{LLN2}
\end{align} 
and
\begin{align}
\bE_V\ppp{\sum_{m=0}^{n-1}\Img\ppp{\tilde{X}_m\tilde{y}_m^*e^{-\frac{2\pi jlm}{n}}}} = n\rho^l_{I},\ l=0,\ldots,k
\label{LLN3}
\end{align} 
where the expectation $\bE_V$ is taken w.r.t. the backward channel $V\p{\cdot\vert\tilde{\by},\btheta,k}$. This parameter vector can be found by solving the set of equations \eqref{alphaEst}-\eqref{alphaEst2}, namely, it attains the maximum of $V\p{\tilde{\bx}\vert\tilde{\by},\btheta,k}$ as can be easily seen. Then,
\begin{align}
1&\geq V\p{\ppp{\calT_\epsilon^k\p{\tilde{\bx}\vert\tilde{\by}}}\vert\tilde{\by},\hat{\btheta},k}\\
& = \int_{\calT_\epsilon^k\p{\tilde{\bxt}\vert\tilde{\byt}}}V(\bar{\bx}\vert\tilde{\by},\hat{\btheta},k)\mathrm{d}\bar{\bx}\\
&\geq \text{Vol}\ppp{\calT_\epsilon^k\p{\tilde{\bx}\vert\tilde{\by}}}\inf_{\bar{\bxt}\in\calT_\epsilon^k\p{\tilde{\bxt}\vert\tilde{\byt}}}V(\bar{\bx}\vert\tilde{\by},\hat{\btheta},k)\\
&\geq \text{Vol}\ppp{\calT_\epsilon^k\p{\tilde{\bx}\vert\tilde{\by}}}\exp\ppp{-n\pp{\frac{1}{2\hat{\sigma}^2}\p{1+2\sum_{l=0}^k\abs{\hat{\alpha}_l}}\epsilon}}V(\tilde{\bx}\vert\tilde{\by},\hat{\btheta},k)\\
&\geq \text{Vol}\ppp{\calT_\epsilon^k\p{\tilde{\bx}\vert\tilde{\by}}}\exp\ppp{-n\epsilon B\pp{1+\sqrt{k+1}\cdot C\p{B^{-1},\Delta}}}V(\tilde{\bx}\vert\tilde{\by},\hat{\btheta},k)
\label{upperboundVolu}
\end{align}
where the second last inequality readily follows from a derivation similar to \eqref{exponenEquiva} and the fact that $\p{\tilde{\bx},\tilde{\by}}\in H_n\p{B}$, and the last inequality follows from \eqref{upperBound} along with the fact that for any sequence $\bz = \p{z_1,\ldots,z_n}$, we have $\norm{\bz}_1\leq\sqrt{n}\norm{\bz}_{\ell_2}$. Thus, we obtain
\begin{align}
\text{Vol}\ppp{\calT_\epsilon^k\p{\tilde{\bx}\vert\tilde{\by}}}\leq \frac{\exp\ppp{n\epsilon f\p{B,\Delta,k}}}{\max_{\thetavecsc}V\p{\tilde{\bx}\vert\tilde{\by},\btheta,k}} = \exp\ppp{n\epsilon f\p{B,\Delta,k}}\exp\ppp{n\log\p{\pi e\hat{\sigma}_0^2}}
\end{align}
For a lower bound on the volume, we first note that 
\begin{align}
1 &= V\p{\ppp{\calT_\epsilon^k\p{\tilde{\bx}\vert\tilde{\by}}\cup\ppp{\calT_\epsilon^k\p{\tilde{\bx}\vert\tilde{\by}}}^c}\vert\tilde{\by},\hat{\btheta},k} \\
&\leq \text{Vol}\ppp{\calT_\epsilon^k\p{\tilde{\bx}\vert\tilde{\by}}}\exp\ppp{n\epsilon f\p{B,\Delta,k}}\max_{\thetavecsc}V\p{\tilde{\bx}\vert\tilde{\by},\btheta,k} + V\p{\ppp{\calT_\epsilon^k\p{\tilde{\bx}\vert\tilde{\by}}}^c\vert\tilde{\by},\hat{\btheta},k}
\label{onestep}
\end{align}
where the last inequality follows by the same considerations in \eqref{upperboundVolu}. Using Boole's inequality
\begin{align}
V\p{\ppp{\calT_\epsilon^k\p{\tilde{\bx}\vert\tilde{\by}}}^c\vert\tilde{\by},\hat{\btheta},k}&\leq V\p{\abs{\frac{1}{n}\sum_{m=0}^{n-1}{\abs{\tilde{X}_m}^2}-\rho_{xx}}>\epsilon\Bigg\vert\tilde{\by},\hat{\btheta},k}\nonumber\\
& \ \ \ + \sum_{l=0}^kV\p{\abs{\frac{1}{n}\sum_{m=0}^{n-1}\re\ppp{\tilde{X}_m\tilde{y}_m^*e^{-\frac{2\pi jlm}{n}}}-\rho^l_{R}}>\epsilon\Bigg\vert\tilde{\by},\hat{\btheta},k}\nonumber\\
& \ \ \ + \sum_{l=0}^kV\p{\abs{\frac{1}{n}\sum_{m=0}^{n-1}\Img\ppp{\tilde{X}_m\tilde{y}_m^*e^{-\frac{2\pi jlm}{n}}}-\rho^l_{I}}>\epsilon\Bigg\vert\tilde{\by},\hat{\btheta},k}
\label{LargeDevi}
\end{align}
Now, due to \eqref{LLN1}-\eqref{LLN3}, the events in \eqref{LargeDevi} are large deviations events. For example, for the second term on the right hand side of \eqref{LargeDevi}, let us define the following Gaussian density
\begin{align}
\delta_G\p{\bz} = \frac{1}{\p{\pi\hat{\sigma}_0^2}^n}\exp\ppp{-\frac{1}{\hat{\sigma}_0^2}\sum_{m=0}^{n-1}\abs{z_m}^2}.
\label{GaussianMeasure2}
\end{align}
Whence, by Chebychev's inequality we obtain, for any $0\leq l\leq k$,
\begin{align}
V\p{\abs{\frac{1}{n}\sum_{m=0}^{n-1}\re\ppp{\tilde{X}_m\tilde{y}_m^*e^{-\frac{2\pi jlm}{n}}}-\rho^l_{xy}}>\epsilon\Bigg\vert\tilde{\by},\hat{\btheta},k} &= \delta_G\ppp{\bZ:\;\abs{\frac{1}{n}\sum_{m=0}^{n-1}\re\ppp{Z_m\tilde{y}_m^*e^{-\frac{2\pi jlm}{n}}}}>\epsilon}\nonumber\\
&\leq \frac{1}{\epsilon^2}\bE_\delta\p{\frac{1}{n}\sum_{m=0}^{n-1}\re\ppp{Z_m\tilde{y}_m^*e^{-\frac{2\pi jlm}{n}}}}^2\\
&\leq\frac{1}{n\epsilon^2}\pp{\sum_{m=0}^{n-1}\abs{\tilde{y}_m}^2}\bE_\delta\ppp{\frac{1}{n}\sum_{m=0}^{n-1}\abs{Z_m}^2}\\
&\leq\frac{B\hat{\sigma}_0^2}{n\epsilon^2}\leq\frac{B^2}{n\epsilon^2}.
\label{kappa2}
\end{align}
For the third term on the right hand side of \eqref{LargeDevi}, we again have that
\begin{align}
V\p{\abs{\frac{1}{n}\sum_{m=0}^{n-1}\Img\ppp{\tilde{X}_m\tilde{y}_m^*e^{-\frac{2\pi jlm}{n}}}-\rho^l_{xy}}>\epsilon\Bigg\vert\tilde{\by},\hat{\btheta},k}&\leq\frac{1}{\epsilon^2}\bE_\delta\p{\frac{1}{n}\sum_{m=0}^{n-1}\Img\ppp{Z_m\tilde{y}_m^*e^{-\frac{2\pi jlm}{n}}}}^2\\
&\leq\frac{B^2}{n\epsilon^2}.
\label{kappa3}
\end{align}
Finally, exactly in the same way, one obtains that
\begin{align}
V\p{\abs{\frac{1}{n}\sum_{m=0}^{n-1}\re\ppp{\abs{\tilde{X}_m}^2e^{\frac{2\pi jlm}{n}}}-\rho_{xx}^l}>\epsilon\Bigg\vert\tilde{\by},\hat{\btheta},k}\leq \frac{12B^2}{n\epsilon^2}
\label{kappa}
\end{align}
Therefore, using \eqref{onestep}, \eqref{kappa2}, \eqref{kappa3}, and \eqref{kappa}, we finally conclude that
\begin{align}
\text{Vol}\ppp{\calT_\epsilon^k\p{\tilde{\bx}\vert\tilde{\by}}}&\geq \frac{\exp\ppp{-n\epsilon f\p{B,\Delta,k}}}{\max_{\thetavecsc}V\p{\tilde{\bx}\vert\tilde{\by},\btheta,k}}\pp{1-\frac{12B^2}{n\epsilon^2}-2k\frac{B^2}{n\epsilon^2}}\\
&\geq\frac{\exp\ppp{-n\epsilon f\p{B,\Delta,k}}}{\max_{\thetavecsc}V\p{\tilde{\bx}\vert\tilde{\by},\btheta,k}}\pp{1-(2k+12)\frac{B^2}{n\epsilon^2}}.
\end{align}
\end{proof}

We are now ready to derive a lower bound on the denominator of \eqref{overBound}. Since $\tilde{\bx}\in\calS_o^\delta\p{\tilde{\bx},\tilde{\by}}$, then, in view of \eqref{exponenEquiva}, there exist a sufficiently small $\epsilon>0$ and a sufficiently large $n$ (both depending on $\delta$) such that $\calT_\epsilon^k\p{\tilde{\bx}\vert\tilde{\by}}\subset\calS_o^\delta\p{\tilde{\bx},\tilde{\by}}$. Thus, using Lemma \ref{lem:Volumes}, we get
\begin{align}
\int_{\calS_o^\delta\p{\tilde{\bxt},\tilde{\byt}}}\mu\p{\bx'}\mathrm{d}\bx'&\geq\int_{\calT_\epsilon^k\p{\tilde{\bxt}\vert\tilde{\byt}}}\mu\p{\bx'}\mathrm{d}\bx'\\
&\geq\text{Vol}\ppp{\calT_\epsilon^k\p{\tilde{\bx}\vert\tilde{\by}}}\cdot\inf_{\bxt'\in\calT_\epsilon^k\p{\tilde{\bxt}\vert\tilde{\byt}}}\mu\p{\bx'}\\
&\geq \frac{\exp\ppp{-n\epsilon f\p{B,\Delta,k}}}{\max_{\thetavecsc}V\p{\tilde{\bx}\vert\tilde{\by},\btheta,k}}\pp{1-(2k+12)\frac{B^2}{n\epsilon^2}}e^{-nC_2\epsilon}\mu\p{\tilde{\bx}}.
\label{eq:1}
\end{align}

We next overbound the numerator of \eqref{overBound}. The basic idea here is to decompose $\calS_u\p{\tilde{\bx},\tilde{\by}}$ into subexponential number of conditional types, where for each conditional type, $\int_{\calT_\epsilon^k\p{\tilde{\bxt}\vert\tilde{\byt}}}\mu\p{\bx'}\mathrm{d}\bx'$ is overestimated using Lemma \ref{lem:Volumes}. Yet, this cannot be done directly, simply because not every $\tilde{\bxt}'\in\calS_u\p{\tilde{\bx},\tilde{\by}}$ is such that $\p{\tilde{\bxt}',\tilde{\by}}\in H_n\p{B}$ and hence we cannot apply Lemma \ref{lem:Volumes} to $\calT_\epsilon^k\p{\tilde{\bxt}'\vert\tilde{\by}}$. Thus, in order to alleviate this difficulty, let us divide $\calS_u\p{\tilde{\bx},\tilde{\by}}$ into two subsets, $\calS_u\p{\tilde{\bx},\tilde{\by}}\cap H_n\p{B_0\vert\tilde{\by}}$ and $\calS_u\p{\tilde{\bx},\tilde{\by}}\cap H_n^c\p{B_0\vert\tilde{\by}}$, where $H_n\p{B_0\vert\tilde{\by}}\define\ppp{\bx':\;\p{\bx',\tilde{\by}}\in H_n\p{B_0}}$, $B_0\geq B$, being a constant to be chosen later. Now, in the first set we can apply Lemma \ref{lem:Volumes} while the second has a very low probability provided that $B_0$ is sufficiently large. Let $B$ be large enough so that \eqref{Hconstranit} holds and fix $\p{\tilde{\bx},\tilde{\by}}\in H_n\p{B}$. Similarly to Lemma \ref{lem:Suff}, one can choose $B_0$ so large such that for every $\by'\in H_n\p{B\vert\tilde{\bx}}$, we have 
\begin{align}
\int_{H_n^c\p{B_0\vert\tilde{\byt}}}\mu\p{\bx'}\mathrm{d}\bx'\leq e^{-nQ\p{B_0}},
\end{align}
for all large $n$, where $Q\p{B_0}>0$ can be made arbitrarily large. Thus, we have
\begin{align}
\int_{\calS_u\p{\bxt,\byt}}\mu\p{\bx'}\mathrm{d}\bx'&\leq \int_{\calS_u\p{\bxt,\byt}\cap H_n\p{B_0\vert\byt}}\mu\p{\bx'}\mathrm{d}\bx' + e^{-nQ\p{B_0}}.
\label{eq:2}
\end{align}
Let us now subdivide the domain of the first term on the r.h.s. of the above inequality into conditional $\epsilon$-types, whose volumes can be overestimated by Lemma \ref{lem:Volumes}. To this end, we will need the number of such sets required to cover the whole domain of integration, that is $\calS_u\p{\bx,\by}\cap H_n\p{B_0\vert\by}\subset H_n\p{B_0\vert\by}$. We note that within this set, $n^{-1}\sum_{m=0}^{n-1}\abs{\tilde{x}'_m}^2\leq B_0$, $n^{-1}\sum_{m=0}^{n-1}\abs{\tilde{y}_m}^2\leq B_0$, and hence also $n^{-1}\abs{\sum_{m=0}^{n-1}\re\ppp{\tilde{x}'_m\tilde{y}_m^*e^{2\pi jlm/n}}}\leq B_0$ and $n^{-1}\abs{\sum_{m=0}^{n-1}\Img\ppp{\tilde{x}'_m\tilde{y}_m^*e^{2\pi jlm/n}}}\leq B_0$ for all $l = 0,\ldots,k$. Thus, the number of conditional types $\ppp{\calT_\epsilon^k\p{\bx'\vert\by}}$ needed to cover $H_n\p{B_0\vert\by}$ is not larger than $\p{2B_0/\epsilon}^{2k+3}$. Therefore,
\begin{align}
\int_{\calS_u\p{\tilde{\bxt},\tilde{\byt}}\cap H_n\p{B_0\vert\tilde{\byt}}}\mu\p{\bx'}\mathrm{d}\bx'&\leq \sum\limits_{\calT_\epsilon^k\p{\tilde{\bxt}'\vert\tilde{\byt}}\subset\calS_u\p{\tilde{\bxt},\tilde{\byt}}\cap H_n\p{B_0\vert\tilde{\byt}}}\int_{\calT_\epsilon^k\p{\tilde{\bxt}'\vert\tilde{\byt}}}\mu\p{\bx''}\mathrm{d}\bx''\\
&\leq\p{\frac{2B_0}{\epsilon}}^{2k+3}\sup_{\tilde{\bxt}'\in\calS_u\p{\tilde{\bxt},\tilde{\byt}}}\ppp{\text{Vol}\ppp{\calT_\epsilon^k\p{\tilde{\bx}'\vert\tilde{\by}}}\cdot\sup_{\bxt''\in\calT_\epsilon^k\p{\tilde{\bxt}'\vert\tilde{\byt}}}\mu\p{\bx''}}\nonumber\\
&\leq \p{\frac{2B_0}{\epsilon}}^{2k+3}\exp\ppp{n\epsilon f\p{B,\Delta,k}}e^{nC_2\epsilon}\sup_{\tilde{\bxt}'\in\calS_u\p{\tilde{\bxt},\tilde{\byt}}}\frac{\mu\p{\tilde{\bx}'}}{\max_{\thetavecsc}V\p{\tilde{\bx}'\vert\tilde{\by},\btheta,k}}\nonumber\\
&\leq \p{\frac{2B_0}{\epsilon}}^{2k+3}\exp\ppp{n\epsilon f\p{B,\Delta,k}}e^{nC_2\epsilon}\frac{\mu\p{\tilde{\bx}}}{\max_{\thetavecsc}V\p{\tilde{\bx}\vert\tilde{\by},\btheta,k}}.
\label{eq:3}
\end{align}
Therefore, combining \eqref{eq:1}, \eqref{eq:2}, and \eqref{eq:3}, we get for all sufficiently large $n$, 
\begin{align}
\sup_{\p{\tilde{\bxt},\tilde{\byt}}\in H_n\p{B}}\frac{\int_{\calS_u\p{\tilde{\bxt},\tilde{\byt}}}\mu\p{\bx'}\mathrm{d}\bx'}{\int_{\calS_0^\delta\p{\tilde{\bxt},\tilde{\byt}}}\mu\p{\bx'}\mathrm{d}\bx'}\leq &\pp{1-(2k+12)\frac{B^2}{n\epsilon^2}}^{-1}\p{\frac{2B_0}{\epsilon}}^{2k+3}e^{2n\epsilon \pp{C_2+f\p{B,\Delta,k}}}\nonumber\\
&\ \cdot\pp{1+e^{-nQ\p{B_0}}\sup_{\p{\tilde{\bxt},\tilde{\byt}}\in H_n\p{B}}\frac{\max_{\thetavecsc}V\p{\tilde{\bx}\vert\tilde{\by},\btheta,k}}{\mu\p{\tilde{\bx}}}}.
\label{ratioSub}
\end{align}
We next provide the conditions under which the last bound is indeed a subexponential function of $n$. To this end, let us first handle the squared brackets in \eqref{ratioSub}, and show it tends to unity as $n\to\infty$ by choosing $Q\p{B_0}$ to be sufficiently large. Note that the supremum can be bounded by
\begin{align}
\sup_{\p{\tilde{\bxt},\tilde{\byt}}\in H_n\p{B}}\frac{\max_{\thetavecsc}V\p{\tilde{\bx}\vert\tilde{\by},\btheta,k}}{\mu\p{\tilde{\bx}}} &= \sup_{\p{\tilde{\bxt},\tilde{\byt}}\in H_n\p{B}}\frac{\p{\pi e\hat{\sigma}_0^2}^{-n/2}}{\mu\p{\tilde{\bx}}}\\
&\leq\frac{\p{\pi eB^{-1}}^{-n/2}}{\nu^{-1}e^{\p{1+\Delta}n/2}},
\end{align}
and that the normalization constant $\nu$ can also be upper bounded as follows 
\begin{align}
\nu = \int_{\bxt\in\Psi_{\Delta}}\mathrm{d}\bx\exp\ppp{-\frac{1}{2\sigma_x^2}\sum_{t=0}^{n-1}x_t^2}\leq e^{-1\p{1-\Delta}n/2}\pp{2\pi e\sigma_x^2\p{1+\Delta}}^{n/2}. 
\end{align}
Whence, using the last results and \eqref{ratioSub}, we see that by choosing $B_0$ so large so that
\begin{align}
Q\p{B_0}> \frac{1}{2}\pp{\log B+\log\sigma_x^2+\log\p{1+\Delta}+2\Delta},
\end{align}
the last term in the squared brackets in \eqref{ratioSub} tends to unity as $n\to\infty$, as required. Thus, in order that \eqref{ratioSub} will be a subexponential function of $n$, we let $\epsilon = \epsilon_n$ tend to zero and $k = k_n$, such that 
\begin{align}
\lim_{n\to\infty}\frac{1}{n}\log\ppp{\pp{1-(2k_n+12)\frac{B^2}{n\epsilon_n^2}}^{-1}\p{\frac{2B_0}{\epsilon_n}}^{2k_n+3}e^{2n\epsilon_n \pp{C_2+f\p{B,\Delta,k_n}}}} = 0,
\label{req}
\end{align}
or, equivalently, that the following hold simultaneously 
\begin{align}
k_n\log\frac{1}{\epsilon_n}=o\p{n},\label{req1}\\
\lim_{n\to\infty}\sqrt{k_n}\epsilon_n = 0,\label{req2}
\end{align}
and
\begin{align}
\lim_{n\to\infty}\frac{n\epsilon_n^2}{k_n} = C_2,\label{req3}
\end{align}
where $C_2$ is some sufficiently large constant, and \eqref{req1}, \eqref{req2}, and \eqref{req3} follow from the midterm, right, and left terms on the left hand side of \eqref{req}. This happens if $\epsilon_n = o\p{n^{-1/4}}$ and hence $k_n = o\p{n^{1/2}}$. Whence, we obtain that \eqref{ratioSub} is subexponential function of $n$, and thus
\begin{align}
\lim_{n\to\infty}\frac{1}{n}\log \bar{P}_{e,u}\p{R,n}\leq \frac{1}{n}\log \bar{P}_{e,0}^\delta\p{R,n},
\end{align}
as required. Finally, to complete the proof of the theorem, it remains to show \eqref{justified}. Note that both $\calS_0\p{\bx,\by}$ and $\calS_0^\delta\p{\bx,\by}$ correspond to a known channel. This is, actually, a similar (and simpler) problem to that we considered above, and is very related to the problem considered in \cite[Eqs. (33)-(39)]{NeriUni}, where \eqref{justified} has been proven. In the sequel, we briefly describe how to obtain \eqref{justified}. Similarly to the above analysis, using Lemma \ref{lem:1}, we would like to show that the ratio
\begin{align}
\frac{\int_{\calS_o^\delta\p{\tilde{\bxt},\tilde{\byt}}}\mu\p{\bx'}\mathrm{d}\bx'}{\int_{\calS_o\p{\tilde{\bxt},\tilde{\byt}}}\mu\p{\bx'}\mathrm{d}\bx'}
\label{asas}
\end{align}
is uniformly overbounded by a subexponential function of $n$, over $\p{\tilde{\bxt},\tilde{\byt}}\in H_n$ where $H_n$ is defined exactly as in \eqref{HnSet}. For a given pair of vectors $\p{\tilde{\bx},\tilde{\by}}$ and $\epsilon > 0$, define the $k$th order conditional $\epsilon$-type $\calT_\epsilon^k\p{\tilde{\bx}\vert\tilde{\by}}$ exactly as in \eqref{TypConCon}. Accordingly, we know that for any $\tilde{\bu},\tilde{\bv}\in\calT_\epsilon^k\p{\tilde{\bx}\vert\tilde{\by}}$ the conditional pdf's $W\p{\by\vert\bu}$ and $W\p{\by\vert\bv}$ are exponentially equivalent, that is, \eqref{exponenEquiva} holds. Then, in view of the last fact, there exists a sufficiently small $\epsilon_1>0$ and a sufficiently large $n$ such that $\calT_\epsilon^k\p{\tilde{\bx}\vert\tilde{\by}}\subset \calS_o\p{\tilde{\bxt},\tilde{\byt}}$, and another $\epsilon_2>0$ and a sufficiently large $n$ (both depending on $\delta$) such that $\calS_o^\delta\p{\tilde{\bxt},\tilde{\byt}}\subset \calT_\epsilon^k\p{\tilde{\bx}\vert\tilde{\by}}$. Then, using the same techniques as previously described, it is possible to overbound the numerator and underbound the denominator of the r.h.s. of \eqref{asas} in terms of the volumes of the conditional types $\calT_\epsilon^k\p{\tilde{\bx}\vert\tilde{\by}}$, and show that \eqref{asas} is overbounded by a subexponential function of $n$.
\appendices
\numberwithin{equation}{section}
\section{Proof of Lemma \ref{lem:3}}
\label{app:1}
We need to show the inclusion
\begin{align}
\calT_\epsilon^k\p{\tilde{\bx}\vert\tilde{\by}}\subseteq\hat{\calT}_\epsilon^k\p{\tilde{\bx}\vert\tilde{\by}},
\end{align}
namely, for any $\bar{\bx}\in\calT_\epsilon^k\p{\tilde{\bx}\vert\tilde{\by}}$ also $\bar{\bx}\in\hat{\calT}_\epsilon^k\p{\tilde{\bx}\vert\tilde{\by}}$. Using the definitions of these sets we see that in order to show the above inclusion we only need to show that for every $\bar{\bx}\in\calT_\epsilon^k\p{\tilde{\bx}\vert\tilde{\by}}$, there exist a sequence $\ppp{P_m}_{m=1}^L\in\boldsymbol{\mathcal{P}}^{\epsilon}$ such that for any $1\leq l\leq L$,
\begin{align}
\abs{\norm{\bar{\bx}_{\p{l-1}n_b+1}^{ln_b}}^2-n_{b}P_l}\leq\epsilon
\label{alocatePP}
\end{align}
where $\bx_l^m\define \p{x_l,x_{l+1},\ldots,x_m}$ for $m\geq l$. To this end, for each $1\leq l\leq L$, $P_l$ is chosen to be the nearest point to $\norm{\bar{\bx}_{\p{l-1}n_b+1}^{ln_b}}^2$ in the set $\calG_{1,\epsilon}^L$, namely $P_l = \left\lfloor \norm{\bar{\bx}_{\p{l-1}n_b+1}^{ln_b}}^2/\p{n_b\epsilon}\right\rfloor\cdot\epsilon$. Under this choice, obviously, \eqref{alocatePP} holds, and $\ppp{P_l}_{l=1}^L\in\boldsymbol{\mathcal{P}}^{\epsilon}$, since
\begin{align}
\abs{\frac{1}{L}\sum_{l=1}^kP_l-P_x} &= \abs{\frac{1}{L}\sum_{l=1}^L\left\lfloor \frac{\norm{\bar{\bx}_{\p{l-1}n_b+1}^{ln_b}}^2}{n_b\epsilon}\right\rfloor\epsilon-P_x}\\
&\leq\abs{\frac{1}{n}\sum_{l=1}^L\norm{\bar{\bx}_{\p{l-1}n_b+1}^{ln_b}}^2\delta-P_x}\leq\delta
\end{align}
where the last equality follows from the fact that $\bar{\bx}\in\calT_\epsilon^k\p{\tilde{\bx}\vert\tilde{\by}}$ and that $n = n_bL$. 
\ifCLASSOPTIONcaptionsoff
  \newpage
\fi
\bibliographystyle{IEEEtran}
\bibliography{strings}
\end{document}